\documentclass[journal]{IEEEtran}
\usepackage[T1]{fontenc}
\usepackage{ifpdf}
\usepackage{cite}
\usepackage[pdftex]{graphicx}
\usepackage{amsmath}
\usepackage{amssymb}
\usepackage{bm}
\usepackage{array}
\usepackage{fixltx2e}
\usepackage{stfloats}

\usepackage{mathtools} 

\usepackage[caption=false,font=footnotesize]{subfig}

\usepackage{diagbox}

\usepackage{makecell}

\usepackage{multirow}

\usepackage{threeparttable}

\usepackage{color,soul}
\soulregister\cite7

\usepackage{amsthm}
\makeatletter
\renewenvironment{proof}[1][\proofname]{\par
  \normalfont
  \topsep6\p@\@plus6\p@ \trivlist
  \item[\hskip3\labelsep\itshape    
    #1\@addpunct{.}]\ignorespaces
}{%
  \qed\endtrivlist
}
\makeatother

\usepackage[hidelinks]{hyperref}

\begin{document}

\theoremstyle{plain} 
\newtheorem{theorem}{Theorem}
\newtheorem{proposition}{Proposition}

\newtheorem{definition}{Definition}

\theoremstyle{remark} 
\newtheorem{remark}{Remark}

\title{\huge On the Computation and Approximation of Outage Probability in Satellite Networks with Smart Gateway Diversity}

\author{Christos~N.~Efrem and~Athanasios~D.~Panagopoulos,~\IEEEmembership{Senior~Member,~IEEE}
\thanks{C. N. Efrem and A. D. Panagopoulos are with the School of Electrical and Computer Engineering, National Technical University of Athens, 15780 Athens, Greece (e-mails: chefr@central.ntua.gr, thpanag@ece.ntua.gr).

This article has been accepted for publication in \textit{IEEE Transactions on Aerospace and Electronic Systems}, DOI: 10.1109/TAES.2020.3022437.  Copyright \textcopyright \ 2020 IEEE. Personal use is permitted, but republication/redistribution requires IEEE permission. See \url{http://www.ieee.org/publications_standards/publications/rights/index.html} for more information.
}}

\markboth{}%
{}

\maketitle

\begin{abstract}
The utilization of extremely high frequency (EHF) bands can achieve very high throughput in satellite networks (SatNets). Nevertheless, the severe rain attenuation at EHF bands imposes strict limitations on the system availability. Smart gateway diversity (SGD) is considered indispensable in order to guarantee the required availability with reasonable cost. In this context, we examine a load-sharing SGD (LS-SGD) architecture, which has been recently proposed in the literature. For this diversity scheme, we define the system outage probability (SOP) using a rigorous probabilistic analysis based on the Poisson binomial distribution (PBD), and taking into consideration the traffic demand as well as the gateway (GW) capacity. Furthermore, we provide several methods for the exact and approximate calculation of SOP. As concerns the exact computation of SOP, a closed-form expression and an efficient algorithm based on a recursive formula are given, both with quadratic worst-case complexity in the number of GWs. Finally, the proposed approximation methods include well-known probability distributions (binomial, Poisson, normal) and a Chernoff bound. According to the numerical results, binomial and Poisson distributions are by far the most accurate approximation methods.

\end{abstract}

\begin{IEEEkeywords}
Satellite networks, EHF bands, feeder link, smart gateway diversity, outage probability, Poisson binomial distribution, recursive formula, approximation methods.
\end{IEEEkeywords}

\IEEEpeerreviewmaketitle

\section{Introduction}
\IEEEPARstart{N}{ext}-generation broadband SatNets require very high data-rates (up to 1 Tbps) that can be accomplished by utilizing EHF bands (above 30 GHz) in the feeder links. Although the frequency shift from Ka (20/30 GHz) to Q/V (40/50 GHz) or W (75-110 GHz) bands provides more spectrum, the  high levels of rain attenuation (tens of dB) cannot be tackled by the standard fade mitigation techniques (FMTs), such as uplink power control (ULPC), adaptive coding and modulation (ACM) and data rate adaptation (DRA). As a result, gateway diversity (GD) is necessary to achieve high system availability, since it is a more effective and powerful FMT (at the expense of installing additional GWs) \cite{Panagopoulos2004, Panagopoulos2005, Kourogiorgas2012, Karagiannis2012, Panagopoulos2016}. Nevertheless, the conventional GD (where the same signal is transmitted by two or three GWs) is economically prohibitive for reaching the Tbps due to the large number of required GWs \cite{Jeannin2014}. \linebreak An alternative solution to achieve high availability with reasonable cost is the smart gateway diversity (SGD), where a user beam can be served by different GWs depending on the propagation conditions and the traffic load. In particular, if a GW experiences deep fades then its traffic can be rerouted to other GWs with better propagation conditions. 

\subsection{Related Work}
In \cite{Jeannin2014}, two SGD techniques are examined, namely, the frequency multiplexing diversity and the $N + P$ diversity. The performance analysis of these schemes is based on a simple probabilistic model, assuming the same outage probability for each GW (although unusual in practice) as well as independent propagation conditions over the GW locations. Moreover, the authors in \cite{Kyrgiazos2014} study the $N$-active diversity (with time or frequency multiplexing, taking into account the spatial correlation between the GWs) and the $N + P$ diversity (where there are $N$ active plus $P$ redundant or idle GWs). In the former scheme, all the $N$ GWs are active and each user beam is served by a group of GWs, whereas in the latter scheme each user beam is served by only one GW and switches to a redundant GW in case of outage. 

A novel GW switching scheme for the $N + P$ scenario is proposed in \cite{Gharanjik2015}, using a dynamic rain attenuation model and considering two key performance indicators: the average outage probability and the average switching rate. Furthermore, a different SGD scheme, where there is no redundant GWs but each GW should have some spare capacity, is analyzed in \cite{Rossi2017}. Specifically, in nominal clear-sky conditions all GWs are active and operate using a maximum fraction of their full capacity, while if some GWs experience heavy rain attenuation then their traffic is served by the remaining GWs using their extra capacity. Finally, an extension of the well-known \linebreak $N$-active and $N + P$ diversity schemes to multiple-input-multiple-output (MIMO) architectures is presented in \cite{Delamotte2020}.

\subsection{Contribution}
The main contributions of this work, in comparison with existing approaches, are as follows:
\begin{itemize}
  \item In this paper, we analyze in detail a SGD architecture operating in \emph{load-sharing mode}, where the GWs do not necessarily have equal outage probabilities. To the best of our knowledge, the concept of LS-SGD has been firstly introduced in\cite{Rossi2017}, assuming that all GWs utilize the same fraction of their full capacity in clear-sky conditions; our analysis, however, does not make such an assumption.
  \item Unlike previous research, we present a \emph{system-level approach} taking into account the \emph{traffic demand} as well as the \emph{GW capacity}. In particular, we are interested in the \emph{system outage probability (SOP)}, defined as the probability of not satisfying the overall traffic demand, which is a \emph{stricter performance metric} than the \emph{user outage probability (UOP)}, i.e., the probability of not satisfying the traffic demand of a specific user.
  \item Furthermore, we study the performance improvement (in terms of SOP) that can be achieved by increasing the number of GWs in the LS-SGD scheme. For this purpose, we define two comparative metrics, namely, the \emph{SOP-improvement factor} and the \emph{generalized SOP-improvement factor}.
  \item In addition, exact methods for the computation of SOP are given, including a \emph{closed-form expression} and an \emph{efficient algorithm based on a recursive formula}. The worst-case complexity of both methods is quadratic in the number of GWs. 
  \item Finally, we provide some approximation methods for the estimation of SOP. More specifically, the SOP can be approximated by various probability distributions (binomial, Poisson, normal) as well as a Chernoff bound.\linebreak Ultimately, we conclude that \emph{binomial} and \emph{Poisson} distributions are the most appropriate approximation methods for SGD systems operating in EHF bands.
\end{itemize}

\subsection{Paper Organization}
The remainder of this article is organized as follows. Firstly, Section \ref{Section_II} describes and analyzes in more detail the LS-SGD architecture. Moreover, Sections \ref{Section_III} and \ref{Section_IV} present exact and approximation methods for calculating the SOP, respectively. In addition, the performance of LS-SGD as well as the accuracy of approximation methods are examined in Section \ref{Section_V}. Finally, concluding remarks are given in Section \ref{Section_VI}.

\subsection{Mathematical Notation \& Conventions}
\textit{Mathematical notation}: ${\mathbb{Z}^ + } = \{ 1,2,3, \ldots \}$, $\mathbb{Z}_0^ +  = \{ 0,1,2, \ldots \}$, $\mathcal{N} = \{ 1,2, \ldots ,N\}$ and $\mathcal{N}{_0} = \{ 0,1, \ldots ,N\}$, where $N \in {\mathbb{Z}^ + }$. Moreover, $\mathbb{P}( \cdot )$ and $\mathbb{E}( \cdot )$ denote probability and expectation, respectively. $\left\lfloor  \cdot  \right\rfloor$ and $\left\lceil  \cdot  \right\rceil$ are respectively the floor and ceiling functions. In addition, $\left| x \right|$ represents the absolute value of a real number $x$, while $\left| \mathcal{S} \right|$ stands for the cardinality of a set $\mathcal{S}$. ${{\mathbf{0}}_N}$ and ${{\mathbf{1}}_N}$ denote the $N$-dimensional all-zeros and all-ones vectors, respectively. Furthermore,\linebreak $\varphi (x) = \left( {\sqrt {2\pi } } \right)^{ - 1}{e^{ - 0.5{x^2}}}$ is the probability density function (PDF), $\Phi (x) = \int_{ - \infty }^x {\varphi (u)} du$ is the cumulative distribution function (CDF), and $Q(x) = 1 - \Phi (x)$ is the complementary CDF (CCDF) of the standard normal distribution. Finally, the \emph{total variation distance} between two (discrete) random variables (RVs) $X$ and $Y$ on $\mathbb{Z}_0^+$ is defined as follows: 
\begin{equation}
\begin{split}
{d_{{\mathrm{TV}}}}(X,Y) & = \mathop {\sup }\limits_{\mathcal{A} \subseteq \mathbb{Z}_0^ + } \left| {\mathbb{P}(X \in \mathcal{A}) - \mathbb{P}(Y \in \mathcal{A})} \right| = \\
 & = \tfrac{1}{2}\sum\limits_{m \in \mathbb{Z}_0^ + } {\left| {\mathbb{P}(X = m) - \mathbb{P}(Y = m)} \right|}
\end{split}
\end{equation}

\textit{Mathematical conventions}: $\sum\limits_{i \in \emptyset }{{a_i}} = 0$ and $\prod\limits_{i \in \emptyset }{{a_i}} = 1$.

\subsection{Preliminaries on Discrete Probability Distributions}

\subsubsection{Bernoulli Distribution}
A binary (0/1) RV follows a Bernoulli distribution with parameter $p \in [0,1]$, $X \sim {\mathrm{Bern}}(p)$, if and only if (iff) its probability mass function (PMF) is given by: $\mathbb{P}(X = 1) = 1 - \mathbb{P}(X = 0) = p$.

\subsubsection{Binomial Distribution} \label{Section_I:binomial}
A discrete (integer-valued) RV $X \sim {\mathrm{Bin}}(N,p)$, where $N \in {\mathbb{Z}^ + }$ and $p \in [0,1]$, iff its PMF is:
\begin{equation}
\mathbb{P}(X = m) = \binom{N}{m} {p^m}{(1 - p)^{N - m}}, \;\; \forall m \in \mathcal{N}{_0}
\end{equation} 
The binomial distribution is a generalization of the Bernoulli distribution, because ${\mathrm{Bin}}(1,p) \equiv {\mathrm{Bern}}(p)$. Furthermore, if ${\{ {X_n}\} _{n \in \mathcal{N}}}$ is a set of \emph{independent and identically distributed (i.i.d.)} Bernoulli RVs $({X_n} \sim {\mathrm{Bern}}(p)$, $\forall n \in \mathcal{N})$, then $S = \sum\limits_{n \in \mathcal{N}} {{X_n}}  \sim {\mathrm{Bin}}(N,p)$.

\subsubsection{Poisson Binomial Distribution} \label{Section_I:Poisson_binomial}
A discrete RV $X \sim {\mathrm{PoisBin}}({\mathbf{p}})$, where ${\mathbf{p}} = [{p_1},{p_2}, \ldots ,{p_N}] \in {[0,1]^N}$ with $N \in {\mathbb{Z}^ + }$, iff its PMF is given by:
\begin{equation}
\mathbb{P}(X = m) = \sum\limits_{\mathcal{A} \in {\mathcal{C}_m}} {\prod\limits_{i \in \mathcal{A}} {{p_i}} \prod\limits_{j \in \mathcal{N}\backslash \mathcal{A}} {(1 - {p_j})} }, \;\; \forall m \in \mathcal{N}{_0}
\end{equation}
where ${\mathcal{C}_m} = \{ \mathcal{A} \subseteq \mathcal{N}:\,\left| \mathcal{A} \right| = m\}$ (i.e., the set of all subsets of $\mathcal{N}$ having $m$ elements) with $\left| {{\mathcal{C}_m}} \right| = \binom{N}{m} = \frac{{N!}}{{m!(N - m)!}}$. The binomial distribution is a special case of the PBD, since ${\mathrm{PoisBin}}(p{{\mathbf{1}}_N}) \equiv {\mathrm{Bin}}(N,p)$. Moreover, if ${\{ {X_n}\} _{n \in \mathcal{N}}}$ is a set of \emph{independent, but not necessarily identically distributed,} Bernoulli RVs $({X_n} \sim {\mathrm{Bern}}({p_n})$, $\forall n \in \mathcal{N})$, then \linebreak $S = \sum\limits_{n \in \mathcal{N}} {{X_n}}  \sim {\mathrm{PoisBin}}({\mathbf{p}})$. 

\subsubsection{Poisson Distribution}
A discrete RV $X \sim {\mathrm{Pois}}(\lambda )$, where $\lambda \geq 0$, iff its PMF is expressed by: $\mathbb{P}(X = m) = {e^{ - \lambda }}{\lambda ^m}{(m!)^{ - 1}}$, $\forall m \in \mathbb{Z}_0^+$.

\section{Smart Gateway Diversity Architecture} \label{Section_II}
In this section, we describe and analyze a \emph{load-sharing SGD (LS-SGD)} scheme, where the unused capacity of available (not in outage) GWs can be exploited to serve the users of the remaining GWs (which are in outage). To the best of our knowledge, this SGD architecture has been firstly proposed and analyzed in\cite{Rossi2017}. Nevertheless, our approach is somewhat different, since it explicitly takes into consideration the \emph{traffic demand} as well as the \emph{GW capacity}.

\subsection{System Model}
Consider a SatNet consisting of a geostationary satellite and a ground network of $N \in {\mathbb{Z}^ +}$ (geographically distributed) GWs, which are denoted by the set $\mathcal{N} = \{ 1,2, \ldots ,N\}$. All the GWs are connected to a \emph{network control center (NCC)} through dedicated terrestrial links. The NCC performs, when necessary (in case of deep fading), the traffic switching/rerouting between the GWs.\footnote{The details on the switching/handover procedure are beyond the scope of this paper; see \cite{Jeannin2014,Gharanjik2015,Muhammad2016} for more information on this important topic.} Furthermore, the following analysis focuses on the \emph{feeder links} (data transmission from the GWs to the satellite), considering ideal (without noise and interference) satellite-user links.\footnote{As concerns the downlink of multibeam satellite systems, an energy-efficient power allocation in order to jointly minimize the unmet system capacity and the total radiated power is proposed in \cite{Efrem2020_power_allocation}.}

In addition, \emph{the distance between any two different GWs is large enough} (some hundreds of km), and thus the spatial correlation of the propagation impairments at the GW locations is extremely small \cite{Jeannin2014, Papafragkakis2019}. As a result, the rain attenuations/fades experienced by the GWs can be considered \emph{(mutually) independent}.  It is also assumed that there is no ACM, so each feeder link is either available at full capacity or completely unavailable.\footnote{Classical FMTs, such as ULPC, ACM and DRA, can tackle impairments of a few dB (e.g., gaseous absorption and cloud attenuation). However, in EHF bands these techniques alone are no longer effective, because the rain attenuation can reach tens of dB. Hence, SGD has to be  used in order to keep SOP at the required levels. In essence, due to the intense rain attenuation in EHF bands, \emph{SGD is the primary FMT}, whereas ULPC, ACM and DRA are secondary/supplementary FMTs. As a result, the absence of ACM in the analysis of SGD is quite reasonable. In any case, our approach provides a \emph{lower bound on the performance} of a more realistic system that utilizes SGD together with standard FMTs.} Therefore, the feeder links can be mathematically modeled as a set ${\{ {X_n}\}_{n \in \mathcal{N}}}$ of \emph{independent, but not necessarily identically distributed,} Bernoulli RVs $({X_n} \sim {\mathrm{Bern}}({p_n})$, $\forall n \in \mathcal{N})$, \linebreak where ${p_n} \in [0,1]$ is the outage/exceedance probability of the ${n^{{\mathrm{th}}}}$ link/GW (i.e., the probability that the rain attenuation exceeds a specific threshold); some methods for calculating ${p_n}$ are discussed in \cite{Rossi2017}. Moreover, we define the RV \linebreak ${S_\mathcal{N}} = \sum\limits_{n \in \mathcal{N}} {{X_n}}  \sim {\mathrm{PoisBin}}({{\mathbf{p}}_\mathcal{N}})$, with ${{\mathbf{p}}_\mathcal{N}} = [{p_1},{p_2}, \ldots ,{p_N}]$, which is \emph{the total number of GWs that are in outage in the set $\mathcal{N}$}.\footnote{According to Section \ref{Section_I:binomial}, if ${p_n} = p$,  $\forall n \in \mathcal{N}$ (i.i.d. Bernoulli RVs), then ${S_\mathcal{N}} \sim {\mathrm{Bin}}(N,p)$. Note that this is rarely the case in practice.}  The expectation, the standard deviation, and the 3\textsuperscript{rd} central moment of ${S_\mathcal{N}}$ are given respectively by:
\begin{equation} \label{mean}
{\mu _\mathcal{N}} = \mathbb{E}({S_\mathcal{N}}) = \sum\limits_{n \in \mathcal{N}} {{p_n}}
\end{equation}
\begin{equation} \label{standard_deviation}
{\sigma _\mathcal{N}} = \sqrt {\mathbb{E}\left( {{{({S_\mathcal{N}} - {\mu _\mathcal{N}})}^2}} \right)}  = \sqrt {\sum\limits_{n \in \mathcal{N}} {{p_n}(1 - {p_n})} }
\end{equation}
\begin{equation} \label{third_central_moment}
{\nu _\mathcal{N}} = \mathbb{E}\left( {{{({S_\mathcal{N}} - {\mu _\mathcal{N}})}^3}} \right) = \sum\limits_{n \in \mathcal{N}} {{p_n}(1 - {p_n})(1 - 2{p_n})}
\end{equation}

Note that ${\mu _\mathcal{N}} \geq \sigma _\mathcal{N}^2$, ${\mu _\mathcal{N}} \in [0,N]$, $\sigma _\mathcal{N}^2 \in [0,N/4]$, and ${\nu _\mathcal{N}} \in [ - N/(6\sqrt 3 ),N/(6\sqrt 3 )]$. 

\subsection{System Outage Probability} \label{Section_II:SOP_definition}
In the sequel, suppose that the ${n^{{\mathrm{th}}}}$ GW can offer a \emph{maximum data-rate (capacity)} $R_n^{{\mathrm{max}}} > 0$, and the \emph{total requested data-rate (traffic demand)} is $R_{{\mathrm{tot}}}^{{\mathrm{req}}} = \sum\limits_{u \in \mathcal{U}} {R_u^{{\mathrm{req}}}} > 0$, where $\mathcal{U} = \{ 1,2, \ldots ,U\}$ is the set of users and $R_u^{{\mathrm{req}}} \geq 0$ is the requested data-rate of user $u$. Moreover, the operation of NCC ensures the following \emph{load-sharing property}: all users receive their requested data-rate if and only if (iff) the overall capacity of the available (not in outage) GWs is greater than or equal to the traffic demand. Equivalently, there is at least one user that receives inadequate data-rate iff the overall capacity of the available GWs is less than the traffic demand. 

\begin{definition}[General SOP expression]
The SOP is defined as follows:
\begin{equation} \label{initial_SOP}
P_{{\mathrm{out}}}^{{\mathrm{sys}}} = \sum\limits_{\mathcal{A} \in \mathcal{F}} {\prod\limits_{i \in \mathcal{A}} {{p_i}} \prod\limits_{j \in \mathcal{N}\backslash \mathcal{A}} {(1 - {p_j})} }
\end{equation}
where $\mathcal{F} = \left\{ {\mathcal{A} \subseteq \mathcal{N}: \, \sum\limits_{j \in \mathcal{N}\backslash \mathcal{A}} {R_j^{{\mathrm{max}}}}  < R_{{\mathrm{tot}}}^{{\mathrm{req}}}} \right\}$. In other words, $\mathcal{F}$ contains all the subsets $\mathcal{A}$ of the $N$ GWs such that: if the GWs in $\mathcal{A}$ are all in outage and the remaining GWs in $\mathcal{N}\backslash \mathcal{A}$ are all available (not in outage), then the traffic demand cannot be satisfied by the latter group of GWs. In essence, the SOP expresses the probability of not satisfying the traffic demand of all users (or, equivalently, the probability that there is at least one user that receives inadequate data-rate). Similarly, we can define the \emph{system availability (SA)} as the probability of the complementary event: $P_{{\mathrm{avail}}}^{{\mathrm{sys}}} = 1 - P_{{\mathrm{out}}}^{{\mathrm{sys}}}$.
\end{definition}

For simplicity, we assume that \emph{all GWs have the same capacity, $R_{{\mathrm{GW}}}^{{\mathrm{max}}} > 0$, in the rest of the paper}; this is not such a strong assumption in practice, since the same frequency band is fully reused in each feeder link and the clear-sky link budget is almost identical for all GWs. 

\begin{theorem}[Special SOP expression]
Suppose that all GWs have the same capacity, i.e., $R_n^{{\mathrm{max}}} = R_{{\mathrm{GW}}}^{{\mathrm{max}}} > 0$, $\forall n \in \mathcal{N}$. Then, \eqref{initial_SOP} reduces to the following expression:
\begin{equation}\label{SOP}
P_{{\mathrm{out}}}^{{\mathrm{sys}}} = P_{{\mathrm{out}}}^{{\mathrm{sys}}}(L,N) = \sum\limits_{m = L}^N {\sum\limits_{\mathcal{A} \in {\mathcal{C}_m}} {\prod\limits_{i \in \mathcal{A}}{{p_i}} \prod\limits_{j \in \mathcal{N}\backslash \mathcal{A}} {(1 - {p_j})}}}
\end{equation}
where ${\mathcal{C}_m} = \{ \mathcal{A} \subseteq \mathcal{N}: \, \left| \mathcal{A} \right| = m\}$ and $L$ is given by:
\begin{equation}\label{L_definition}
L = N - \left\lceil r \right\rceil + 1
\end{equation}
where $r > 0$ is \emph{the ratio of the traffic demand to the GW capacity}, that is:
\begin{equation}\label{ratio}
r = {R_{\mathrm{tot}}^{\mathrm{req}}}/{R_{\mathrm{GW}}^{\mathrm{max}}} 
\end{equation}
\end{theorem}

\begin{proof}
Under the condition of equal GW capacities, we have that $\mathcal{F} = \left\{ {\mathcal{A} \subseteq \mathcal{N}: \, (N - \left| \mathcal{A} \right|)R_{{\mathrm{GW}}}^{{\mathrm{max}}} < R_{{\mathrm{tot}}}^{{\mathrm{req}}}} \right\}$. In addition, $(N - \left| \mathcal{A} \right|)R_{{\mathrm{GW}}}^{{\mathrm{max}}} < R_{{\mathrm{tot}}}^{{\mathrm{req}}}$ $\Leftrightarrow$ $N - \left| \mathcal{A} \right| < r$ $\Leftrightarrow$ $N - \left| \mathcal{A} \right| < \left\lceil r \right\rceil$ $\Leftrightarrow$ $N - \left| \mathcal{A} \right| \leq \left\lceil r \right\rceil - 1$ $\Leftrightarrow$ $\left| \mathcal{A} \right| \geq N - \left\lceil r \right\rceil + 1$. Consequently, $\mathcal{F} = \left\{ {\mathcal{A} \subseteq \mathcal{N}: \, \left| \mathcal{A} \right| \geq L} \right\} = \bigcup\limits_{m = L}^N {{\mathcal{C}_m}}$ and then \eqref{SOP} follows immediately from \eqref{initial_SOP}.
\end{proof}

\begin{remark} \label{rmrk_1}
According to Section \ref{Section_I:Poisson_binomial}, $P_{{\mathrm{out}}}^{{\mathrm{sys}}}(L,N) = \sum\limits_{m = L}^N {\mathbb{P}({S_\mathcal{N}} = m)}  = \mathbb{P}({S_\mathcal{N}} \geq L)$, i.e., the SOP is \emph{the probability of having at least $L$ out of $N$ GWs in outage}.\footnote{Similar formula is also given in \cite{Rossi2017} and \cite{Rossi2020}, however, without explicit dependence on the traffic demand and the GW capacity. Herein, this dependence is clearly expressed by \eqref{L_definition} and \eqref{ratio}. Note that this SOP definition is a generalization of \emph{the classical SOP} (i.e., the probability of having all GWs in outage), which is obtained when $\left\lceil r \right\rceil = 1$ $\Rightarrow$ $L = N$ $\Rightarrow$ $P_{{\mathrm{out}}}^{{\mathrm{sys}}} = \prod\limits_{n \in \mathcal{N}}{p_n}$; the classical SOP is used in\cite{Efrem2020_site_diversity} to select the (globally) minimum number of GWs satisfying SOP-requirements.}
\end{remark}

Although in general $L \in {\mathcal{N}_0}$, for the diversity system under consideration $L \in \mathcal{N}$ due to the fact that $\left\lceil r \right\rceil  \in \mathcal{N}$, since a) $r > 0$ $\Leftrightarrow$ $\left\lceil r \right\rceil \geq 1$, and b) $N R_{{\mathrm{GW}}}^{{\mathrm{max}}} \geq R_{{\mathrm{tot}}}^{{\mathrm{req}}}$ $\Leftrightarrow$ $N \geq r$ $\Leftrightarrow$ $N \geq \left\lceil r \right\rceil$ (note that ${N_{\min }} = \left\lceil r \right\rceil$ is \emph{the minimum required number of GWs}). Finally, we provide a result about the monotonicity of SOP.

\begin{proposition}[SOP monotonicity] \label{prop_1}
For a given set $\mathcal{N}$ of GWs, the SOP is an increasing function of $r$.
\end{proposition}

\begin{proof}
Let ${r_1} \geq {r_2}$ $\Rightarrow$ $\left\lceil {{r_1}} \right\rceil  \geq \left\lceil {{r_2}} \right\rceil$ $\Rightarrow$ ${L_1} \leq {L_2}$ $\Rightarrow$ $P_{{\mathrm{out}}}^{{\mathrm{sys}}}({L_1},N) \geq P_{{\mathrm{out}}}^{{\mathrm{sys}}}({L_2},N)$.
\end{proof}

\subsection{SOP-Improvement Factor} \label{Section_II:SOP_improvement_factor}
Subsequently, we study the performance improvement (in terms of SOP) achieved by an $N$-GW diversity system in comparison with a single-GW system. 

\begin{definition}[SOP-improvement factor]
Assuming the same $\left\lceil r \right\rceil = 1$ and that $P_{{\mathrm{out}}}^{{\mathrm{sys}}}(N,N) > 0$, the SOP-improvement factor is defined as follows:
\begin{equation}\label{improvement_factor}
I = \frac{{P_{{\mathrm{out}}}^{{\mathrm{sys}}}(1,1)}}{{P_{{\mathrm{out}}}^{{\mathrm{sys}}}(N,N)}} = \frac{{{p_1}}}{{\prod\limits_{n \in \mathcal{N}} {{p_n}} }} = {\left( {\prod\limits_{n = 2}^N {{p_n}} } \right)^{ - 1}} 
\end{equation}
Obviously, it holds that $I \geq 1$. 
\end{definition}

Next, consider a diversity system with $N + K$ GWs $(K \in \mathbb{Z}_0^+)$ all of which have the same capacity $R_{{\mathrm{GW}}}^{{\mathrm{max}}} > 0$, and $\left\lceil r \right\rceil  \in \mathcal{N}$ $($since $1 \leq \left\lceil r \right\rceil  \leq \min (N,N + K) = N)$. Furthermore, let $\mathcal{K} = \{ N + 1,N + 2, \ldots ,N + K\}$ be the set of additional GWs, and ${{\mathbf{p}}_{\mathcal{N} \cup \mathcal{K}}} = [{{\mathbf{p}}_\mathcal{N}},{{\mathbf{p}}_\mathcal{K}}] = [{p_1},{p_2}, \ldots ,{p_{N + K}}]$ be the vector of GW outage probabilities, where ${{\bf{p}}_\mathcal{K}} = [{p_{N + 1}},{p_{N + 2}}, \ldots ,{p_{N + K}}]$. Suppose also that ${\{ {X_i}\} _{i \in \mathcal{N} \cup \mathcal{K}}}$ is a set of \emph{independent, but not necessarily identically distributed,} Bernoulli RVs $({X_i} \sim {\mathrm{Bern}}({p_i})$, $\forall i \in \mathcal{N} \cup \mathcal{K})$. Besides ${S_\mathcal{N}}$, we define the RVs ${S_\mathcal{K}} = \sum\limits_{k \in \mathcal{K}} {{X_k}}  \sim {\mathrm{PoisBin}}({{\bf{p}}_\mathcal{K}})$ and ${S_{\mathcal{N} \cup \mathcal{K}}} = \sum\limits_{i \in \mathcal{N} \cup \mathcal{K}} {{X_i}}  = {S_\mathcal{N}} + {S_\mathcal{K}} \sim {\mathrm{PoisBin}}({{\mathbf{p}}_{\mathcal{N} \cup \mathcal{K}}})$ denoting the total number of GWs which are in outage in the sets $\mathcal{K}$ and $\mathcal{N} \cup \mathcal{K}$, respectively. For this diversity system $L' = N + K - \left\lceil r \right\rceil  + 1 = L + K$, with $L' \in \{ K + 1,K + 2, \ldots ,K + N\}$.  

\begin{proposition}[SOP reduction] \label{prop_2}
Let $P_{{\mathrm{out}}}^\mathcal{N} = P_{{\mathrm{out}}}^{{\mathrm{sys}}}(L,N) = \mathbb{P}({S_\mathcal{N}} \geq L)$ and $P_{{\mathrm{out}}}^{\mathcal{N} \cup \mathcal{K}} = P_{{\mathrm{out}}}^{{\mathrm{sys}}}(L',N + K) = \mathbb{P}({S_{\mathcal{N} \cup \mathcal{K}}} \geq L')$ stand for the SOP of the $N$-GW and $(N + K)$-GW diversity systems, respectively. Then, it holds that $P_{{\mathrm{out}}}^{\mathcal{N} \cup \mathcal{K}} \leq P_{{\mathrm{out}}}^\mathcal{N}$.
\end{proposition}

\begin{proof}
See Appendix \ref{appndx_A}.
\end{proof}

In view of this fact, we can generalize the definition of SOP-improvement factor. 

\begin{definition}[Generalized SOP-improvement factor]
Assuming the same $\left\lceil r \right\rceil  \in \mathcal{N}$ and that $P_{{\mathrm{out}}}^{\mathcal{N} \cup \mathcal{K}} > 0$, we define the generalized SOP-improvement factor of the $(N + K)$-GW over the $N$-GW diversity system as follows:\footnote{The generalized SOP-improvement factor $I_{\mathrm{g}}$ can be estimated using the approximation methods provided in Section \ref{Section_IV}.} 
\begin{equation}\label{generalized_improvement_factor}
I_{\mathrm{g}} = \frac{{P_{{\mathrm{out}}}^\mathcal{N}}}{{P_{{\mathrm{out}}}^{\mathcal{N} \cup \mathcal{K}}}} = {\left. {\frac{{P_{{\mathrm{out}}}^{{\mathrm{sys}}}(L,N)}}{{P_{{\mathrm{out}}}^{{\mathrm{sys}}}(L + K,N + K)}}} \right|_{L = N - \left\lceil r \right\rceil  + 1}} 
\end{equation}
According to Proposition \ref{prop_2}, we have that $I_{\mathrm{g}} \geq 1$.
\end{definition} 

Notice that by setting $N = 1$ and $K = N' - 1$ $($thus $\left\lceil r \right\rceil  = 1$ and $L = 1)$, we obtain $I_{\mathrm{g}} = \frac{{P_{{\mathrm{out}}}^{{\mathrm{sys}}}(1,1)}}{{P_{{\mathrm{out}}}^{{\mathrm{sys}}}(N',N')}} = I$. Finally, we would like to emphasize that by increasing the number of GWs the SOP decreases, but higher GW connectivity is required; such connectivity issues are very important in the design and optimization of SatNets\cite{Trufero2014}. In other words, there is a trade-off between performance improvement and connectivity complexity.

\section{Exact Methods for Computing SOP} \label{Section_III}
In the sequel, several techniques for the exact computation of SOP are presented. The time complexity of these methods is summarized in Table \ref{Table_1}.

\begin{table}[!t]
\caption{Complexity Comparison Between Exact Methods}
\centering
\renewcommand{\arraystretch}{2.3}
\scalebox{0.88}{\begin{tabular}{|c|c|c|c|c|}
\hline
\makecell{\textbf{Exact} \\ \textbf{Method}} & \makecell{Direct \\ Computation} & CFE & \makecell{RF \\ (Algorithm 1)} & \makecell{FFT-based \\ Algorithm \cite{Belfore1995}} \\ \hline
\makecell{\textbf{Time} \\ \textbf{Complexity}} & $O({2^N}N)$ & $\Theta({N^2})$ & \makecell{$\Theta (L(N - L + 1))$ \\ $= O({N^2})$} & $O(N{(\log N)^2})$ \\ \hline
\end{tabular}}
\label{Table_1}
\end{table}

\subsection{Direct Computation}
The direct computation of SOP is based on the analytic formula \eqref{SOP}, which requires $\sum\limits_{m = L}^N {\left| {{\mathcal{C}_m}} \right|N} = N\sum\limits_{m = L}^N {\binom{N}{m}} \leq N\sum\limits_{m = 0}^N {\binom{N}{m}} = {2^N}N = O({2^N}N)$ arithmetic operations. Because of its exponential worst-case complexity, this method is practicable only for very small $N$.

\subsection{Closed-Form Expression}
According to \cite{Fernandez2010}, the SOP can be calculated, using polynomial interpolation and discrete Fourier transform (DFT), by the following \emph{closed-form expression (CFE)}: 
\begin{equation}\label{CFE}
\scalebox{0.9}{$P_{{\mathrm{out}}}^{{\mathrm{sys}}}(L,N) = 1 - \tfrac{1}{{N + 1}}\left( {L + \sum\limits_{n \in \mathcal{N}} {\tfrac{{1 - {c^{ - nL}}}}{{1 - {c^{ - n}}}}\prod\limits_{m \in \mathcal{N}} {\left(1 + ({c^n} - 1){p_m}\right)} } } \right)$}
\end{equation}
where $c = {e^{j2\pi /(N + 1)}}$, with $j = \sqrt { - 1}$ being the imaginary unit. It is interesting to note that the CFE comprises a sum of complex numbers, but the overall outcome is a real number in $[0,1]$. The same formula is also derived in \cite{Hong2013}, using the characteristic function of the PBD as well as the DFT. Furthermore, the computational complexity of \eqref{CFE} is $\Theta({N^2})$.

\subsection{Recursive Formula}
In this part, we explore the power and beauty of recursion. 

\begin{theorem}[SOP recursive formula] \label{thrm_RF}
The SOP is given by the following \textit{recursive formula (RF)}:
\begin{equation}\label{RF}
\scalebox{0.93}{$P_{{\mathrm{out}}}^{{\mathrm{sys}}}(L,N) = (1 - {p_N})P_{{\mathrm{out}}}^{{\mathrm{sys}}}(L,N - 1) + {p_N}P_{{\mathrm{out}}}^{{\mathrm{sys}}}(L - 1,N - 1)$}
\end{equation}
with initial/boundary conditions: a) $P_{{\mathrm{out}}}^{{\mathrm{sys}}}(0,N) = 1$ and b) $P_{{\mathrm{out}}}^{{\mathrm{sys}}}(N + 1,N) = 0$, $\forall N \in {\mathbb{Z}^ + }$.
\end{theorem}

\begin{proof}
See Appendix \ref{appndx_B}. 
\end{proof}

It can be verified, using mathematical induction, that \eqref{SOP} is the solution of \eqref{RF}. To the best of our knowledge, this RF is derived for the first time in \cite{Rushdi1986}, making use of symmetric switching functions. Our proof, however, is much simpler. 

Algorithm 1 presents an efficient method to compute the SOP using the RF, which follows directly from the algorithm given in \cite{Rushdi1986}. The \emph{time complexity} of Algorithm 1 is \linebreak $\Theta (L(N - L + 1)) = O({N^2})$, with best-case complexity $\Theta (1)$ for $L = 0$, and worst-case complexity $\Theta ({N^2})$ for $L = \left\lfloor  {N/2}  \right\rfloor$ and $L = \left\lceil  {N/2}  \right\rceil$. Moreover, notice that the complexity is $\Theta (N)$ for $L = 1$ and $L = N$. As a result, Algorithm 1 has lower complexity in some cases than the CFE which requires $\Theta ({N^2})$ operations regardless of $L$. Finally, the \emph{space complexity} of Algorithm 1 is $\Theta (N + L) = \Theta (N)$.

\begin{table}[!t]
\centering
\renewcommand{\arraystretch}{1.35}
\begin{tabular*}{\columnwidth}{@{}l@{}}
\hline
\textbf{\normalsize{Algorithm 1}} \normalsize{Exact computation of SOP}
\\ \hline
\textbf{Input:} $N \in {\mathbb{Z}^ + }$, $L \in {\mathcal{N}_0}$, and ${\mathbf{p}} = [{p_1},{p_2}, \ldots ,{p_N}] \in {[0,1]^N}$ \\
\textbf{Output:} $P_{{\mathrm{out}}}^{{\mathrm{sys}}} = P_{{\mathrm{out}}}^{{\mathrm{sys}}}(L,N)$ \\
\begin{tabular}{@{}r@{~}l@{}}
1: & $D \coloneqq N - L$, $M \coloneqq L + 1$, ${\boldsymbol{\alpha }} \coloneqq {{\mathbf{0}}_M}$, ${\alpha _1} \coloneqq 1$, $\ell \coloneqq 1$ \\
2: & \textbf{for} $i \coloneqq 1$ \textbf{to} $N$ \textbf{step} $+1$ \textbf{do} \\
3: & ~~~$h \coloneqq i$  \\ 
4: & ~~~\textbf{if} $i > D + 1$ \textbf{then} $\ell \coloneqq i - D$ \textbf{end if}  \\
5: & ~~~\textbf{if} $i > L$ \textbf{then} $h \coloneqq L$ \textbf{end if} \\
6: & ~~~\textbf{for} $j \coloneqq h$ \textbf{to} $\ell$ \textbf{step} $-1$ \textbf{do}   \hfill $\triangleright$ \textit{$h$,$\ell$: high/low index} \\
7: & ~~~~~~${\alpha _{j + 1}} \coloneqq (1 - {p_i}) \cdot {\alpha _{j + 1}} + {p_i} \cdot {\alpha _j}$ \\
8: & ~~~\textbf{end for} \\
9: & \textbf{end for} \\
10: & $P_{{\mathrm{out}}}^{{\mathrm{sys}}} \coloneqq {\alpha _M}$ \\    
\end{tabular}
\\ \hline
\end{tabular*}
\end{table}

\subsection{FFT-based Algorithm}
An even more efficient and advanced algorithm for computing the SOP is provided in \cite{Belfore1995}. This method recursively applies the fast Fourier transform (FFT) to compute \emph{generating function (GF)} products, thus achieving an overall complexity of $O(N{(\log N)^2})$. 

In particular, the PMF of ${S_\mathcal{N}}  \sim {\mathrm{PoisBin}}({{\mathbf{p}}_\mathcal{N}})$ can be written in the following form:
\begin{equation}
\begin{gathered}
\left[\mathbb{P}(S_\mathcal{N}=0) \;\; \mathbb{P}(S_\mathcal{N}=1) \;\; \cdots \;\; \mathbb{P}(S_\mathcal{N}=N) \right] = \hfill \\
= \left[ {q_1} \;\; {p_1} \right] * \left[ {q_2} \;\; {p_2} \right] * \cdots * \left[ {q_N} \;\; {p_N} \right]  \hfill \\
\end{gathered}
\end{equation}
where $*$ stands for the \emph{convolution} operation and ${q_n} = 1-{p_n}$, $\forall n \in \mathcal{N}$. In addition, the GF of the Poisson-binomial PMF is given by:
\begin{equation}
\begin{split}
g(z) & = \sum\limits_{n \in {\mathcal{N}_0}} {\mathbb{P}(S_\mathcal{N}=n) \, z^n} = \prod\limits_{n \in \mathcal{N}} {({q_n} + {p_n} z)} =  \\
& = g_{\pi}  \prod\limits_{n \in \mathcal{N}} {(1 + {a_n} z)} =  g_{\pi} \left(1 + A(z)\right) 
\end{split}
\end{equation}
where $g_{\pi} = \prod\limits_{n \in \mathcal{N}} {q_n} $ and $a_{n} = {p_n} / {q_n} $, $\forall n \in \mathcal{N}$. Obviously, the SOP is simply the sum of the coefficients of $z^m$ from $m=L$ to $N$ (see Remark \ref{rmrk_1}). Since the product of two GF is equivalent to the convolution of two sequences formed from the GF coefficients, the FFT can be used to compute GF products more efficiently compared to the term-by-term calculation. The basic idea of the algorithm proposed in \cite{Belfore1995} is to apply the FFT to compute the GF $A(z)$ using a \emph{divide-and-conquer} approach. More details on the implementation of the algorithm can be found therein.

\begin{remark}
Despite the fact that the FFT-based algorithm is more sophisticated and has lower asymptotic complexity, CFE and Algorithm 1 are sufficient in the context of SGD, where the number of GWs $N$ is relatively small.
\end{remark}

\section{Approximation Methods for Estimating SOP} \label{Section_IV}
Afterwards, we introduce some useful methods to \linebreak approximate the SOP, exploiting the fact that $P_{{\mathrm{out}}}^{{\mathrm{sys}}}(L,N) = \linebreak \mathbb{P}({S_\mathcal{N}} \geq L) = 1 - \mathbb{P}({S_\mathcal{N}} \leq L - 1)$, $\forall L \in {\mathcal{N}_0}$. These techniques consist of \emph{probability distributions (binomial, Poisson, normal)} as well as a \emph{Chernoff bound}. For convenience, a summary of approximation methods is given in Table \ref{Table_2}.

\begin{table*}[!t]
\centering
\begin{threeparttable}[b]
\caption{Summary of Approximation Methods}
\centering
\renewcommand{\arraystretch}{2.3}
\begin{tabular}{|c|c|c|c|}
\hline 
\textbf{Approximation Method} & \textbf{SOP Approximation Formula} $\widetilde P_{{\mathrm{out}}}^{{\mathrm{sys}}}(L,N)$ & \textbf{Parameters/Range of} $L$ & \makecell{\textbf{Condition for} \\ \textbf{Higher Accuracy}} \\  \hline 
Binomial Approximation (BA)\tnote{a} & $1 - \sum\limits_{m = 0}^{L - 1} {\binom{N}{m}{{\bar p}^m}{\bar q}^{N - m}}$ & $\bar p = \frac{1}{N}\sum\limits_{n \in \mathcal{N}} {{p_n}} $, \ $\bar q = 1 - \bar p$ & ${(N\bar p\bar q)^{ - 1}}\sigma _\mathcal{N}^2 \to 1$ \\ \hline
Poisson Approximation (PA)\tnote{b} & $1 - {e^{ - {\mu _\mathcal{N}}}}\sum\limits_{m = 0}^{L - 1} {\mu _\mathcal{N}^m{(m!)^{ - 1}}}$ & ---  & $\sum\limits_{n \in \mathcal{N}} {p_n^2} \to 0$ \\ \hline
Normal Approximation (NA) & $1 - \Phi \left( \zeta  \right) = Q\left( \zeta  \right)$ & \multirow{2}{*}{$\zeta  = (L - {\mu _\mathcal{N}} - 0.5)\sigma _\mathcal{N}^{ - 1}$} & \multirow{2}{*}{$\sigma _\mathcal{N}^2 \to \infty$} \\ \cline{1-2} 
Refined Normal Approximation (RNA) & $\min \left(\max \left( 1 - G( \zeta ),0 \right),1 \right)$ & & \\ \hline
Chernoff Bound (CB) & ${({\mu _\mathcal{N}}/L)^L}{e^{L - {\mu _\mathcal{N}}}}$ & $\forall L \in \left\{ {\left\lfloor {{\mu _\mathcal{N}}} \right\rfloor  + 1,\left\lfloor {{\mu _\mathcal{N}}} \right\rfloor + 2, \ldots ,N} \right\}$  &  --- \\ \hline
\end{tabular}
\begin{tablenotes}
\item[a,b] According to the numerical results (Section \ref{Section_V}), BA and PA are the most appropriate approximation methods for SGD systems operating in EHF bands. 
\end{tablenotes}
\label{Table_2}
\end{threeparttable}
\end{table*}

\subsection{Binomial Approximation (BA)} \label{Section_IV:BA}
The PBD can be approximated by the binomial distribution \cite{Ehm1991} in the following sense, defining \linebreak $\bar p = \frac{1}{N}\sum\limits_{n \in \mathcal{N}} {{p_n}} $, $\bar q = 1 - \bar p$,  and assuming $\bar p \in (0,1)$: \linebreak a) ${d_{{\mathrm{TV}}}}({S_\mathcal{N}},Y) \leq (N/(N + 1))(1 - {\bar p^{N + 1}} - {\bar q^{N + 1}}){\delta _\mathcal{N}}$, where $Y \sim {\mathrm{Bin}}(N,\bar p)$ and ${\delta _\mathcal{N}} = 1 - {(N\bar p\bar q)^{ - 1}}\sigma _\mathcal{N}^2$, and b) ${d_{{\mathrm{TV}}}}({S_\mathcal{N}},Y) \to 0$ \emph{if and only if (iff)} ${\delta _\mathcal{N}} \to 0$ $($or, equivalently, ${(N\bar p\bar q)^{ - 1}}\sigma _\mathcal{N}^2 \to 1)$. It is interesting to note that when ${p_n} = p$, $\forall n \in \mathcal{N}$, it holds that: $\bar p = p$, $\bar q = 1 - p$ and $\sigma _\mathcal{N}^2 = N\bar p\bar q$ $\Rightarrow$ ${\delta _\mathcal{N}} = 0$ $\Rightarrow$ ${d_{{\mathrm{TV}}}}({S_\mathcal{N}},Y) = 0$ $\Rightarrow$ ${S_\mathcal{N}} \sim {\mathrm{Bin}}(N,p)$, which is in agreement with Section \ref{Section_I:binomial}. Hence, the BA is given by:
\begin{equation}
\scalebox{0.93}{$
P_{{\mathrm{out}}}^{{\mathrm{sys}}}(L,N) \approx 1 - \mathbb{P}(Y \leq L - 1) = 1 - \sum\limits_{m = 0}^{L - 1} {\binom{N}{m}{{\bar p}^m}{(1 - \bar p)^{N - m}}}$}
\end{equation}

\subsection{Poisson Approximation (PA)}
In 1960, Le Cam \cite{LeCam1960} established a remarkable inequality: ${d_{{\mathrm{TV}}}}({S_\mathcal{N}},Z) \leq \sum\limits_{n \in \mathcal{N}} {p_n^2}$, where $Z \sim {\mathrm{Pois}}({\mu _\mathcal{N}})$. It is obvious that if $\sum\limits_{n \in \mathcal{N}} {p_n^2} \to 0$, then ${d_{{\mathrm{TV}}}}({S_\mathcal{N}},Z) \to 0$. As reported in \cite{Steele1994}, \emph{Le Cam's theorem/inequality} admits various proofs using different techniques. Consequently, we have that:
\begin{equation}
\scalebox{0.93}{$
P_{{\mathrm{out}}}^{{\mathrm{sys}}}(L,N) \approx 1 - \mathbb{P}(Z \leq L - 1) = 1 - {e^{ - {\mu _\mathcal{N}}}}\sum\limits_{m = 0}^{L - 1} {\mu _\mathcal{N}^m{(m!)^{ - 1}}}$}
\end{equation}

\subsection{Normal Approximation (NA)}
According to \cite{Deheuvels1989}, \emph{the central limit theorem (CLT)} for the PBD states that: $\mathop {\lim }\limits_{N \to \infty } {\Delta _\mathcal{N}} = 0$ $($asymptotic normality of $({S_\mathcal{N}} - {\mu _\mathcal{N}})\sigma _\mathcal{N}^{ - 1})$ \emph{iff} $\mathop {\lim }\limits_{N \to \infty } \sigma _\mathcal{N}^2 = \infty $, where ${\Delta _\mathcal{N}} = \mathop {\sup}\limits_{s \in \mathbb{R}} \left| {\mathbb{P}({S_\mathcal{N}} \leq s) - \Phi \left( {(s - {\mu _\mathcal{N}})\sigma _\mathcal{N}^{ - 1}} \right)} \right|$. Therefore, by applying a \emph{continuity correction},\footnote{In probability theory, a \emph{continuity correction} is an adjustment that is made when a discrete (probability) distribution is approximated by a continuous distribution. In particular, suppose that the continuous RV $Y$ approximates the discrete RV $X$. Then, $\mathbb{P}(X \leq m) = \mathbb{P}(X \leq m + 0.5) \approx \mathbb{P}(Y \leq m + 0.5)$, $\forall m \in \mathbb{Z}$.} the SOP can be approximated by: 
\begin{equation}
P_{{\mathrm{out}}}^{{\mathrm{sys}}}(L,N) \approx 1 - \Phi ( \zeta ) = Q( \zeta )
\end{equation}
where $\zeta  = (L - {\mu _\mathcal{N}} - 0.5)\sigma _\mathcal{N}^{ - 1}$.

\subsection{Refined Normal Approximation (RNA)}
Consider the following function:
\begin{equation}\label{G_function}
G(x) = \Phi (x) + {\nu _\mathcal{N}}{(6\sigma _\mathcal{N}^3)^{ - 1}}(1 - {x^2})\varphi (x)
\end{equation}
According to \cite{Deheuvels1989, Mikhailov1994, Volkova1996}, there exists a constant $C < \infty $ such that ${\Delta '_\mathcal{N}} = \mathop {\sup }\limits_{s \in \mathbb{R}} \left| {\mathbb{P}({S_\mathcal{N}} \leq s) - G\left( {(s - {\mu _\mathcal{N}})\sigma _\mathcal{N}^{ - 1}} \right)} \right| \leq C\sigma _\mathcal{N}^{ - 2} = O(\sigma _\mathcal{N}^{ - 2})$. Observe that $\mathop {\lim }\limits_{N \to \infty } {\Delta '_\mathcal{N}} = 0$, when $\mathop {\lim }\limits_{N \to \infty } \sigma _\mathcal{N}^2 = \infty$. As a result, by applying the \emph{continuity correction} once more, we obtain the following approximation: 
\begin{equation}
P_{{\mathrm{out}}}^{{\mathrm{sys}}}(L,N) \approx \min \left(\max \left(\widehat P_{{\mathrm{out}}}^{{\mathrm{sys}}}(L,N),0 \right),1 \right)
\end{equation}
where $\widehat P_{{\mathrm{out}}}^{{\mathrm{sys}}}(L,N) = 1 - G( \zeta )$ and $\zeta  = (L - {\mu _\mathcal{N}} - 0.5)\sigma _\mathcal{N}^{ - 1}$. Note that we make use of the above min-max formula in order to ensure that $P_{{\mathrm{out}}}^{{\mathrm{sys}}}(L,N) \in [0,1]$, because $\widehat P_{{\mathrm{out}}}^{{\mathrm{sys}}}(L,N)$ may be outside the interval $[0,1]$ in some cases.

\subsection{Chernoff Bound (CB)}
A Chernoff (upper) bound can be constructed using a result given in \cite{Hagerup1990} which states that: $\mathbb{P}\left({S_\mathcal{N}} \geq (1 + \delta ){\mu _\mathcal{N}}\right) \leq {\left( {e^\delta }/{(1 + \delta )^{1 + \delta }} \right)^{{\mu _\mathcal{N}}}}$, $\forall \delta > 0$. Specifically, by setting $(1 + \delta ){\mu _\mathcal{N}} = L$ and assuming ${\mu _\mathcal{N}} > 0$, we obtain: 
\begin{equation}
P_{{\mathrm{out}}}^{{\mathrm{sys}}}(L,N) \leq {({\mu _\mathcal{N}}/L)^L}{e^{L - {\mu _\mathcal{N}}}} 
\end{equation}
which holds $\forall L \in \left\{ {\left\lfloor {{\mu _\mathcal{N}}} \right\rfloor  + 1,\left\lfloor {{\mu _\mathcal{N}}} \right\rfloor + 2, \ldots ,N} \right\}$, since \linebreak $\delta  > 0$ $ \Leftrightarrow $ $L > {\mu _\mathcal{N}}$ $ \Leftrightarrow $ $L > \left\lfloor {{\mu _\mathcal{N}}} \right\rfloor $ $ \Leftrightarrow $ $L \geq \left\lfloor {{\mu _\mathcal{N}}} \right\rfloor  + 1$.

\section{Numerical Results and Discussion} \label{Section_V}

In this section, all results present statistical averages derived from ${10^3}$ independent system configurations, where the GW outage probabilities ${\{ {p_i}\} _{i \in {\mathcal{N} \cup \mathcal{K}}}}$ are uniformly distributed in $(0,0.02)$, i.e., 98\% to 100\% link availability.

\subsection{SOP Analysis}

\begin{figure}[!t]
\centering
\includegraphics[width=3.4in]{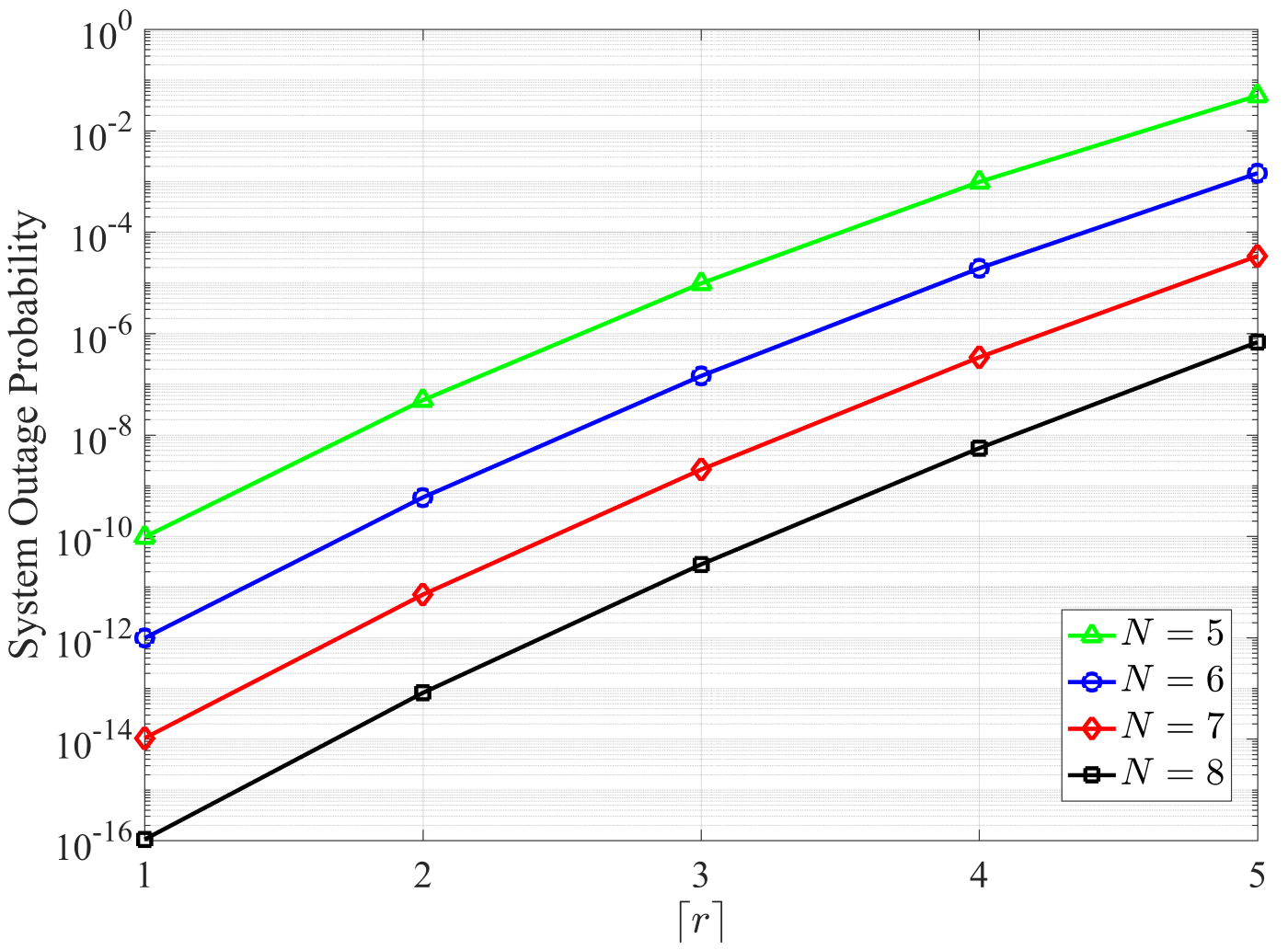} 
\caption{System outage probability, $P_{{\mathrm{out}}}^{{\mathrm{sys}}}$, (calculated using Algorithm 1) versus the ceiling of $r$ (the ratio of the traffic demand to the GW capacity).}
\label{Fig1}
\end{figure}

Firstly, we study the SOP as a function of the number of GWs, $N$, and the ratio of the traffic demand to the GW capacity, $r$. As shown in Fig. \ref{Fig1}, the SOP increases with $\left\lceil r \right\rceil$ for all values of $N$, which is in accordance with Proposition \ref{prop_1}. Moreover, for any fixed $\left\lceil r \right\rceil$, we can observe that the SOP decreases with the increase of $N$ (see Proposition \ref{prop_2}). Nevertheless, as mentioned at the end of Section \ref{Section_II:SOP_improvement_factor}, this SOP improvement is achieved in exchange for higher connectivity complexity.

\begin{figure}[!t]
\centering
\includegraphics[width=3.37in]{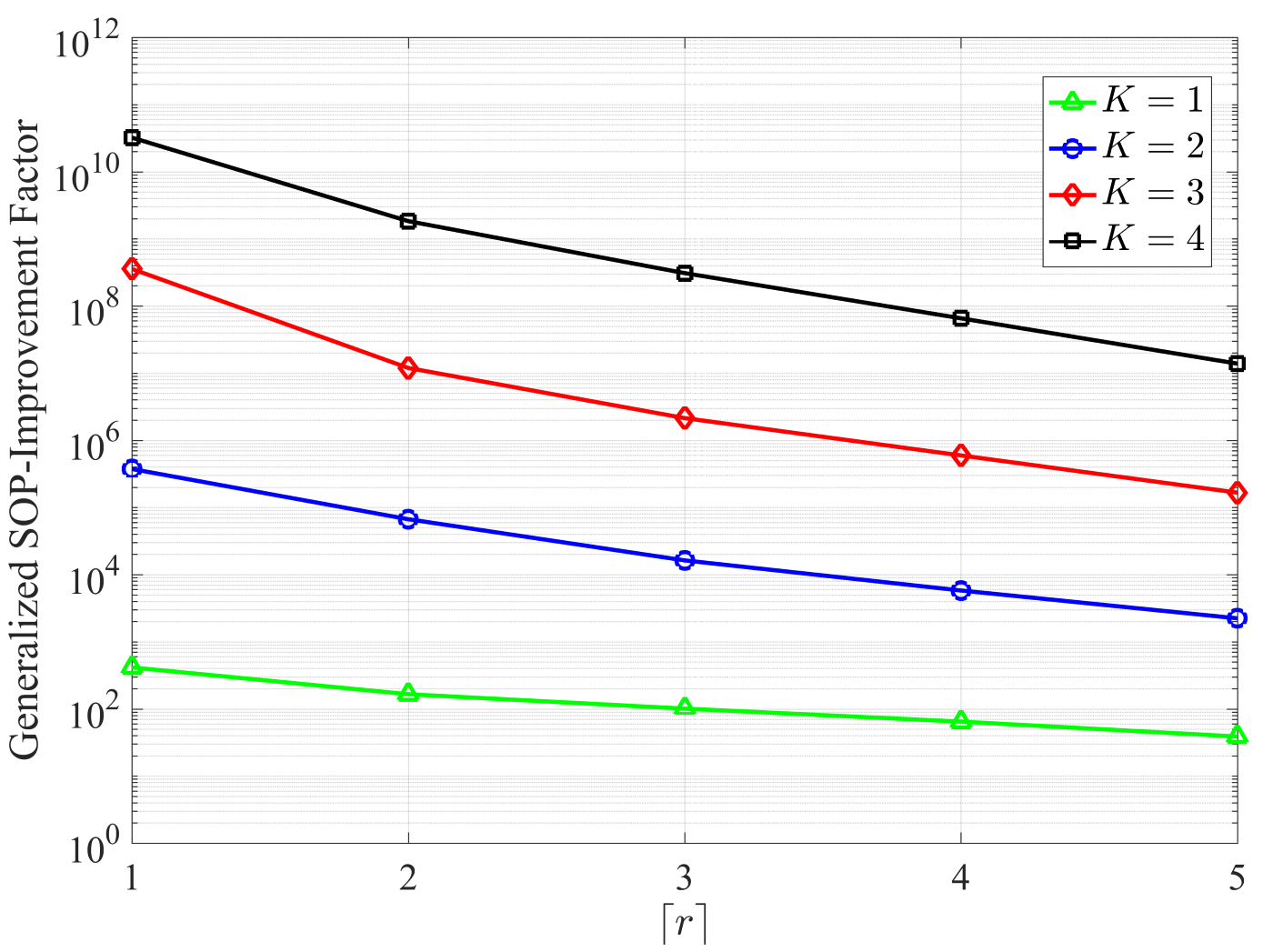} 
\caption{Generalized SOP-improvement factor, $I_{\mathrm{g}}$, (computed using Algorithm 1), in comparison with a diversity system consisting of $N=5$ GWs, versus the ceiling of $r$ (the ratio of the traffic demand to the GW capacity).}
\label{Fig2}
\end{figure}

Secondly, we examine the performance enhancement achieved by a $(5+K)$-GW compared to a $5$-GW diversity system by means of the generalized SOP-improvement factor $($where $K \in \{1,2,3,4\}$ is the number of additional GWs$)$. Specifically, as illustrated in Fig. \ref{Fig2}, $I_{\mathrm{g}}$ decreases with the increase of $\left\lceil r \right\rceil$ for every value of $K$. Furthermore, for a given $\left\lceil r \right\rceil$, larger number of additional GWs results in higher performance improvement.

\subsection{Performance of Approximation Methods}

In order to evaluate the accuracy of a probability distribution and the tightness/sharpness of the Chernoff bound, we define the maximum absolute error (maxAE), the root-mean-square error (RMSE), and the mean absolute error (MAE) as follows: 
\begin{equation}
{{\epsilon} _{\mathrm{max}}}(N) = \mathop {\max }\limits_{L \in \mathcal{S}} {\left| {P_{{\mathrm{out}}}^{{\mathrm{sys}}}(L,N) - \widetilde P_{{\mathrm{out}}}^{{\mathrm{sys}}}(L,N)} \right|}
\end{equation}
\begin{equation}
{{\epsilon} _{\mathrm{rms}}}(N) = \sqrt{\tfrac{1}{{\left|\mathcal{S}\right|}}\sum\limits_{L \in {\mathcal{S}}} {\left( {P_{{\mathrm{out}}}^{{\mathrm{sys}}}(L,N) - \widetilde P_{{\mathrm{out}}}^{{\mathrm{sys}}}(L,N)} \right)^2}}
\end{equation}
\begin{equation}
{{\epsilon} _{\mathrm{mean}}}(N) = \tfrac{1}{{\left|\mathcal{S}\right|}}\sum\limits_{L \in {\mathcal{S}}} {\left| {P_{{\mathrm{out}}}^{{\mathrm{sys}}}(L,N) - \widetilde P_{{\mathrm{out}}}^{{\mathrm{sys}}}(L,N)} \right|}
\end{equation}
where $\widetilde P_{{\mathrm{out}}}^{{\mathrm{sys}}}(L,N)$ is the approximate SOP. Moreover, for probability distributions $\mathcal{S} = \mathcal{N}_0$ $($with $\left|\mathcal{S}\right| = N + 1)$, while for CB $\mathcal{S} = \left\{ {\left\lfloor {{\mu _\mathcal{N}}} \right\rfloor  + 1,\left\lfloor {{\mu _\mathcal{N}}} \right\rfloor + 2, \ldots ,N} \right\}$ $($with $\left|\mathcal{S}\right| = N - \left\lfloor {\mu _\mathcal{N}} \right\rfloor \geq 1)$. In general, it holds that ${{\epsilon} _{\mathrm{max}}}(N) \geq {{\epsilon} _{\mathrm{rms}}}(N) \geq {{\epsilon} _{\mathrm{mean}}}(N)$.

\begin{figure}[!t]
\centering
\subfloat[]{
\includegraphics[width=3.2in]{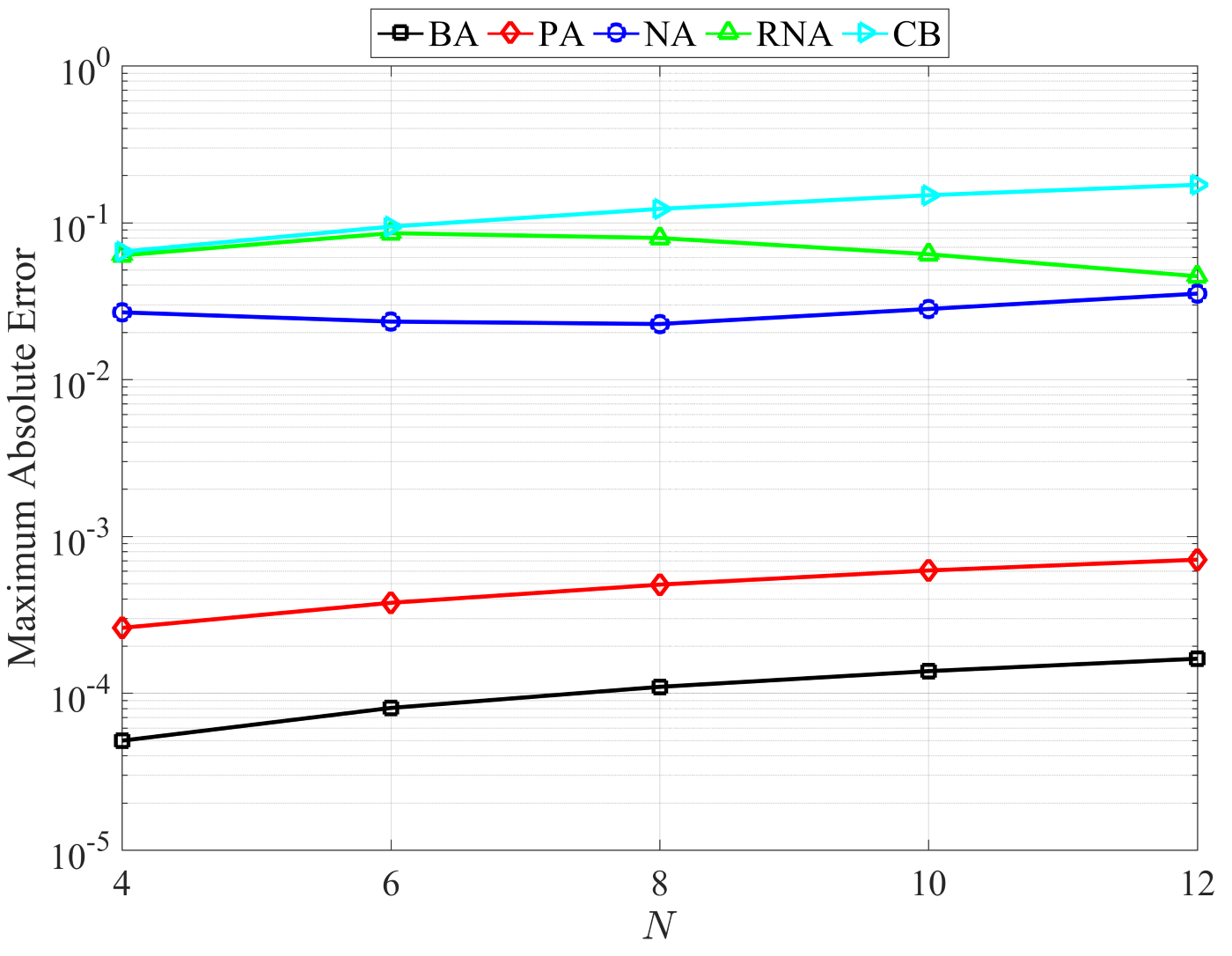} \label{Fig3a} } 
\hfill
\subfloat[]{
\includegraphics[width=3.2in]{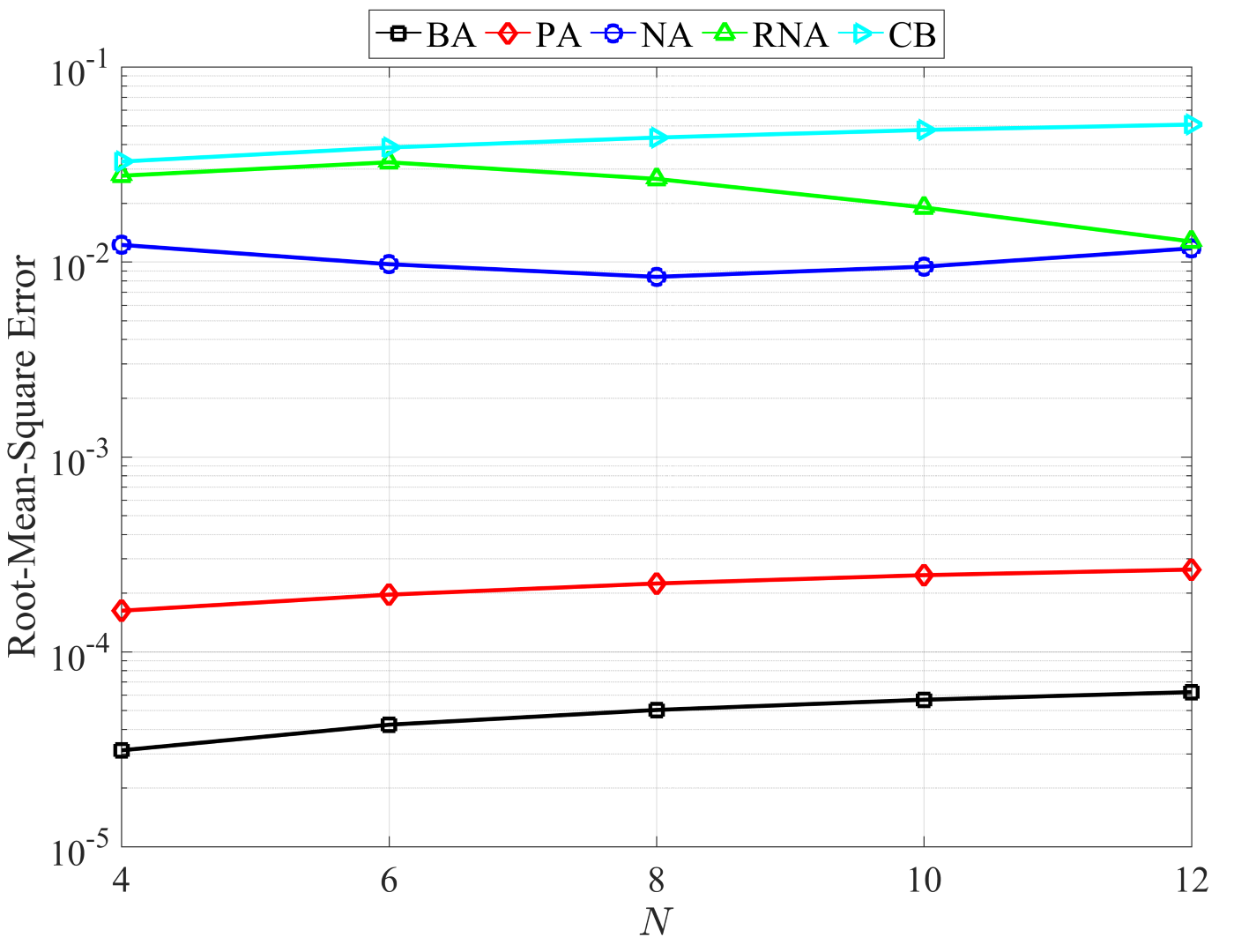} \label{Fig3b} } 
\hfill
\subfloat[]{
\includegraphics[width=3.2in]{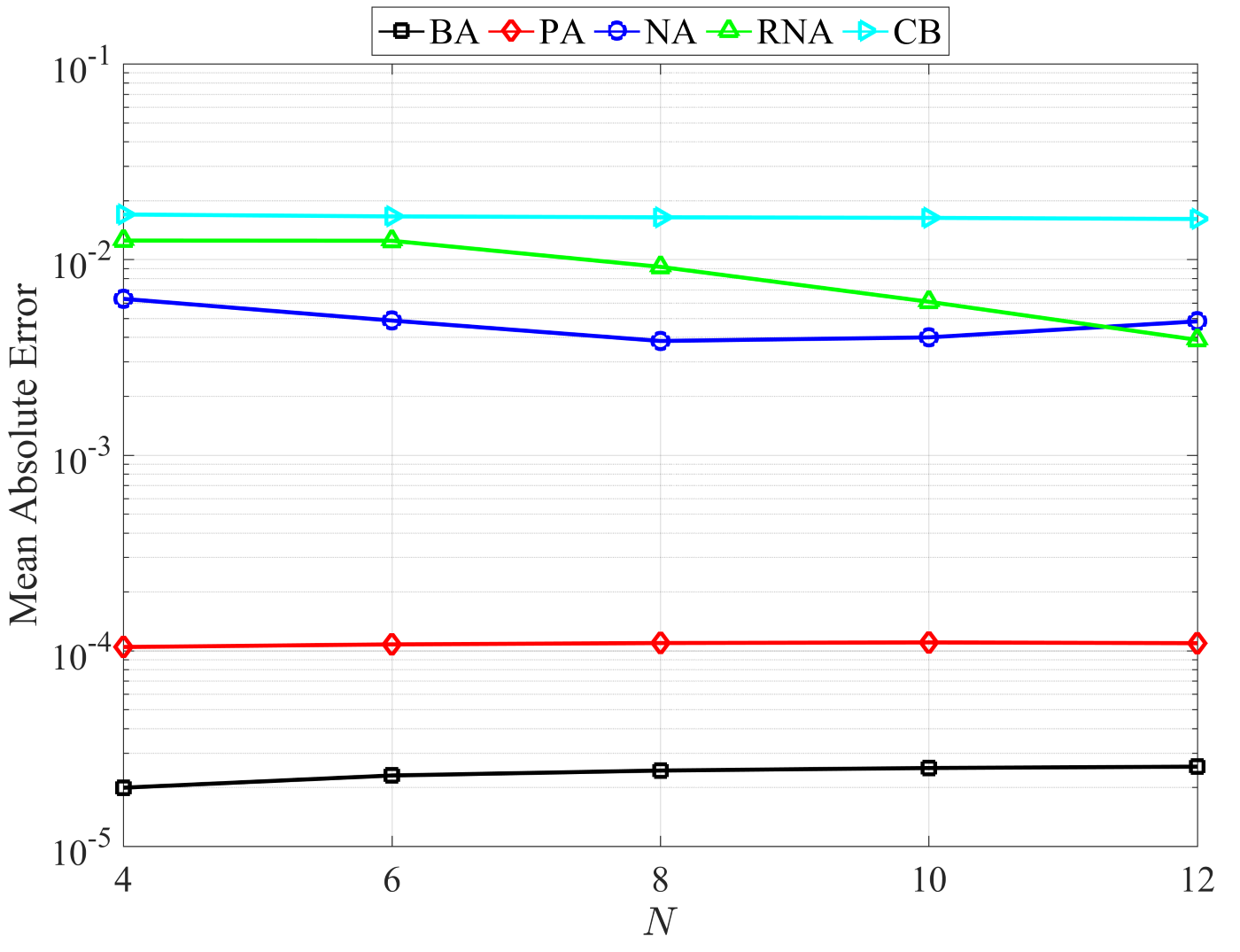} \label{Fig3c} } 
\caption{Accuracy comparison of approximation methods: (a) maximum absolute error, (b) root-mean-square error, and (c) mean absolute error versus the number of GWs.}
\label{Fig3}
\end{figure}

Fig. \ref{Fig3} presents the accuracy of approximation methods, in terms of maxAE, RMSE and MAE, versus the number of GWs. It can be observed that the approximation methods in descending-performance (or, equivalently, ascending-error) order are as follows: $\{$BA, PA, NA, RNA, CB$\}$. More specifically, BA and PA significantly outperform the other methods $($the achieved errors are of the order of $10^{-4}$ or $10^{-5})$, while CB exhibits the lowest accuracy. At this point, we would like to give an explanation of the performance of BA, PA, NA and RNA. In practice, \emph{the number of GWs is relatively small} ($N \approx 4\ \mathrm{to}\ 7$) and \emph{all the GW outage probabilities are very close to zero} (i.e., $p_n \approx 0$, $\forall n \in \mathcal{N}$ $\Rightarrow$ $p_1 \approx p_2 \approx \cdots \approx p_N$). As a result, the variance $\sigma _\mathcal{N}^2 = \sum\limits_{n \in \mathcal{N}} {{p_n}(1 - {p_n})}$ and the quantity $\sum\limits_{n \in \mathcal{N}} {p_n^2}$ are quite small, while $\sigma _\mathcal{N}^2 \approx N\bar p\bar q$ (see Section \ref{Section_IV:BA}). Finally, according to Table \ref{Table_2}, it is clear that the condition for higher accuracy of BA/PA is well satisfied, whereas that of NA/RNA is not. In summary, BA and PA are the most suitable approximation methods for SGD systems.

\section{Conclusion} \label{Section_VI}
In this paper, we have studied in depth the LS-SGD scheme, which has been recently introduced in SatNets. Furthermore, a number of useful mathematical tools have been presented in order to compute and approximate the SOP. Finally, based on the numerical results, we conclude that the SOP can be well approximated by BA and PA, since these methods achieve remarkable accuracy. Such approximations may be useful for simplifying and solving hard optimization problems with SOP-constraints in SGD-based SatNets.

\clearpage

\appendix[]

\subsection{Proof of Proposition \ref{prop_2}} \label{appndx_A}
By virtue of \emph{the law/theorem of total probability}, we obtain: 
\begin{equation}
\scalebox{0.91}{$
\begin{gathered}
  P_{{\mathrm{out}}}^{\mathcal{N} \cup \mathcal{K}} = \mathbb{P}({S_\mathcal{N}} + {S_\mathcal{K}} \geq L + K) =  \hfill \\
   = \sum\limits_{j = 0}^K {\mathbb{P}({S_\mathcal{K}} = j)\mathbb{P}({S_\mathcal{N}} + {S_\mathcal{K}} \geq L + K|{S_\mathcal{K}} = j)}  =  \hfill \\
   = \sum\limits_{j = 0}^K {\mathbb{P}({S_\mathcal{K}} = j)\mathbb{P}({S_\mathcal{N}} \geq L + K - j)}  =  \hfill \\
   = \sum\limits_{j = 0}^K {\mathbb{P}({S_\mathcal{K}} = j)\left[ {\mathbb{P}({S_\mathcal{N}} \geq L) - \mathbb{P}(L \leq {S_\mathcal{N}} \leq L + K - j - 1)} \right]}  \leq  \hfill \\
   \leq \mathbb{P}({S_\mathcal{N}} \geq L)\underbrace {\sum\limits_{j = 0}^K {\mathbb{P}({S_\mathcal{K}} = j)} }_{ = 1} = \mathbb{P}({S_\mathcal{N}} \geq L) = P_{{\mathrm{out}}}^\mathcal{N} \hfill \\
\end{gathered}$}
\end{equation}
and the proposition follows.

\subsection{Proof of Theorem \ref{thrm_RF}} \label{appndx_B}
Firstly, the initial conditions of the RF are trivially true. Secondly, from \emph{the law/theorem of total probability}, the SOP $P_{{\mathrm{out}}}^{{\mathrm{sys}}}(L,N) = \mathbb{P}({S_\mathcal{N}} \geq L)$ can be written as follows:
\begin{equation}
\scalebox{0.92}{$\begin{gathered}
  P_{{\mathrm{out}}}^{{\mathrm{sys}}}(L,N) = \sum\limits_{j = 0}^1 {\mathbb{P}({X_N} = j)\mathbb{P}({S_\mathcal{N}} \geq L|{X_N} = j)}  =  \hfill \\
   = \sum\limits_{j = 0}^1 {\mathbb{P}({X_N} = j)\mathbb{P}({S_{\mathcal{N}\backslash \{ N\} }} \geq L - j)}  =  \hfill \\
   = \mathbb{P}({X_N} = 0)\mathbb{P}({S_{\mathcal{N}\backslash \{ N\} }} \geq L) + \mathbb{P}({X_N} = 1)\mathbb{P}({S_{\mathcal{N}\backslash \{ N\} }} \geq L - 1) \hfill \\ 
\end{gathered}$}
\end{equation}
where ${S_{\mathcal{N}\backslash \{ N\} }} = \sum\limits_{n \in \mathcal{N}\backslash \{ N\} } {{X_n}}  = {S_\mathcal{N}} - {X_N}$. Due to the fact that $\mathbb{P}({X_N} = 0) = 1-{p_N}$ and $\mathbb{P}({X_N} = 1) = {p_N}$, we get \eqref{RF} and this completes the proof.


\end{document}